\documentclass[conference]{IEEEtran}


\usepackage{graphicx,colordvi,psfrag}
\usepackage{amsmath,amssymb}
\usepackage{epstopdf}
\usepackage[caption=false]{subfig}
\usepackage{epsfig,cite}
\usepackage{calc,pstricks, pgf, xcolor}
\usepackage{nicefrac}
\usepackage{enumerate}

\newtheorem{theorem}{Theorem}
\newtheorem{corollary}{Corollary}
\newtheorem{remark}{Remark}

\newtheorem{proposition}{Proposition}
\newtheorem{lemma}{Lemma}

\newenvironment{proof}[1][Proof]{\noindent\textbf{#1.} }{\ \rule{0.5em}{0.5em}}

\newcommand{\bc}{\mathbf{c}}

\def\wt{\widetilde}

\def\b{\mathbf}
\newcommand{\m}{\mathcal}

\def\wt{\widetilde}

\begin{document}

\title{A VC-dimension-based Outer Bound on the Zero-Error Capacity of the Binary Adder Channel}
\author{\IEEEauthorblockN{Or Ordentlich}
\IEEEauthorblockA{Tel Aviv University\\
ordent@eng.tau.ac.il}
\and
\IEEEauthorblockN{Ofer Shayevitz}
\IEEEauthorblockA{Tel Aviv University\\
ofersha@eng.tau.ac.il}
\IEEEoverridecommandlockouts
\IEEEcompsocitemizethanks{
\IEEEcompsocthanksitem
The work of O. Ordentlich was supported by the Admas Fellowship Program of the Israel Academy of Science and Humanities, a fellowship from The Yitzhak and Chaya Weinstein Research Institute for Signal Processing at Tel Aviv University, and the Feder Family Award. The work of O. Shayevitz was supported in part by the Marie Curie Career Integration Grant (CIG), Grant agreement no. 631983, and in part by the Israel Science Foundation under Grant No.  1367/14.
}}

\parskip 3pt

\maketitle

\begin{abstract}
The binary adder is a two-user multiple access channel whose inputs are binary and whose output is the real sum of the inputs. While the Shannon capacity region of this channel is well known, little is known regarding its zero-error capacity region, and a large gap remains between the best inner and outer bounds. In this paper, we provide an improved outer bound for this problem. To that end, we introduce a soft variation of the Saur-Perles-Shelah Lemma, that is then used in conjunction with an outer bound for the Shannon capacity region with an additional common message.
\end{abstract}

\section{Introduction}\label{sec:intro}
The binary adder is a multiple access channel with two binary inputs $X_1$ and $X_2$ and output $Y=X_1+X_2\in\{0,1,2\}$. The capacity region of this channel is well known and consists of all rate-pairs $(R_1,R_2)$ satisfying
\begin{align}
\nonumber R_1 &\leq 1, \\
\nonumber R_2 &\leq 1, \\
R_1+R_2 &\leq \tfrac{3}{2}. \label{eq:ShannonCap}
\end{align}
The zero-error capacity region of the binary adder channel is the closure of the set of all rate-pairs $(R_1,R_2)$ such that for $n$ large enough there exist two codebooks $\m{C}_1,\m{C}_2\subseteq \{0,1\}^n$ with cardinalities $|\m{C}_i|=2^{n R_i}$, $i=1,2$, such that all elements in the \textit{sumset}
\begin{align}
\m{C}_1+\m{C}_2\triangleq\{\b{a}+\b{b} : \b{a}\in\m{C}_1,\b{b}\in\m{C}_2\}\quad \textrm{with multiplicities}\label{eq:sumset}
\end{align}
appear with multiplicity exactly one, where addition is taken over the reals. We say that the pair $(R_1,R_2)$ is \emph{admissible} if it belongs to the zero-error capacity region, and we call the codebooks $(\m{C}_1,\m{C}_2)$ a \emph{zero-error codebook pair} if all elements in their sumset $\m{C}_1+\m{C}_2$ appear with multiplicity exactly one.

Despite its apparent simplicity, the problem of characterizing the zero-error capacity region of this channel is wide open. Many inner bounds have been established over the last four decades, see, e.g.,~\cite{Lindstrom69,Tilborg78,kl78,Weldon78,klwy83,bt85,bb98,UL98,ab99,mo05}. However, to date, the best known lower bound on the zero error sum-capacity is $\log(240/6)\approx 1.3178$~\cite{mo05}, where logarithms are taken in base $2$. To put this result in perspective, note that a sum-rate of $R_1+R_2=\tfrac{1}{2}\log(6)\approx 1.2924$ can be attained by the two-dimensional construction $\m{C}_1=\{00,11\}$, $\m{C}_2=\{00,01,10\}$. In terms of outer bounds, the current state of knowledge is even less satisfying. Clearly, any admissible pair must be inside the Shannon capacity region and must therefore satisfy~\eqref{eq:ShannonCap}. However, to date the only improvement upon the trivial outer bound~\eqref{eq:ShannonCap} was obtained by Urbanke and Li~\cite{UL98} who showed that near the corner points $(1,\tfrac{1}{2})$ and $(\tfrac{1}{2},1)$ the zero-error capacity region is strictly contained in~\eqref{eq:ShannonCap}. Specifically, for $R_1=1$ it was shown that the maximal admissible $R_2$ must satisfy $R_2<0.49216$. Our main result is a new outer bound on the zero-error capacity region that strictly improves upon the bound from~\cite{UL98}.

%The entropy of a random variable $X$ with a probability distribution $P=(p_1,\ldots,p_K)$ is denoted by $H(X)$ or $H(P)$.
Write $h(p) = -p\log{p}-(1-p)\log{(1-p)}$ for the binary entropy function, and $h^{-1}(x)$ for its inverse restricted to $[0,\tfrac{1}{2}]$. For $0\leq p,q \leq 1$, write $p\star q \triangleq p(1-q) + q(1-p)$. Let
\begin{align}\label{eq:L}
L(\eta) &\triangleq h(\eta) + 1- \eta 
\end{align}
and 
\begin{align}\label{eq:J}
J(p,\eta) &\triangleq \left\{\begin{array}{cc}
                                                                                             2h\left(\frac{1}{2}\left(1-\sqrt{1-2\eta}\right)\right)-\eta & \eta\geq p\star p \\ \\ 
                                                                                             2h\left(\frac{1}{2}\left(1-\frac{1-\eta-p\star p}{\sqrt{1-2 (p\star p)}}\right)\right) & \\
                                                                                             -\frac{1}{2}\left(1-\frac{(1-\eta-p\star p)^2}{1-2 (p\star p)} \right) & \eta<p\star p
                                                                                           \end{array}
 \right.
  \end{align}
and
\begin{align}\label{eq:Rsum}
R_\Sigma(r_0,r_1) \triangleq \max_{h^{-1}(r_1)\leq \eta\leq \frac{1}{2}} \min\{L(\eta),\, J(h^{-1}(r_1),\eta)+r_0\}
\end{align}

Our main result is the following.
\begin{theorem}\label{thm:main}
Any admissible $(R_1,R_2)$ satisfies
\begin{align*}
  R_2 < \min_{0\leq \alpha \leq h^{-1}(R_1)}(1-\alpha) \left(R_\Sigma\left(\frac{\alpha}{1-\alpha},\,\Gamma\right) - \Gamma\right)
\end{align*}
where
\begin{align*}
\Gamma = \Gamma(R_1,\alpha) \triangleq h\left(\frac{h^{-1}(R_1)-\alpha}{1-\alpha}\right)
\end{align*}
\end{theorem}
\begin{figure}[htbp]
  \centering
      \includegraphics[width=0.9\columnwidth]{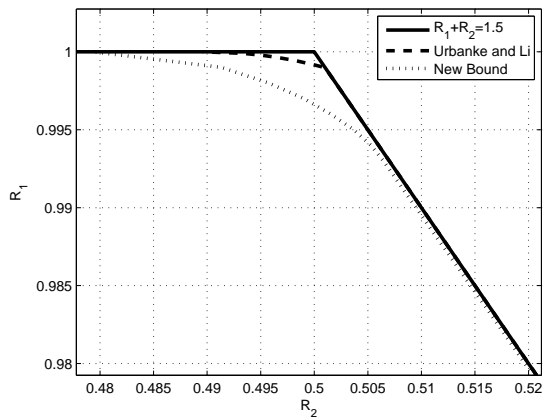}
  \caption{Illustration of the three outer bounds.}\label{fig1}
\end{figure}
For the maximal value of $R_1=1$, this bound yields $R_2<0.4794$. Figure \ref{fig1} depicts the three outer bounds for values of $R_1$ close to $1$. The question of whether $R_1+R_2 = \tfrac{3}{2}$ is admissible for some $(R_1,R_2)$ remains open.

\section{Proof of Theorem \ref{thm:main}}

We first note that it suffices to prove inadmissibility in the limit of large $n$, by the simple fact that if $(\m{C}_1,\m{C}_2)$ is a zero-error codebook pair, so is the concatenation $(\m{C}_1\times\m{C}_1,\m{C}_2\times\m{C}_2)$. To avoid cumbersome notations, we can therefore assume without loss of generality that $nR_1$ and $nR_2$ (and all similar quantities) are integers.

\subsection{Motivation}
Let $\m{C}\subseteq\{0,1\}^n$ be a codebook and let $S\subseteq [n]$ be a subset of coordinates, where $[n] \triangleq\{1,\ldots,n\}$. The projection $\b{a}(S)$ maps the vector $\b{a}\in\{0,1\}^n$ to a vector in $\{0,1\}^{|S|}$ by taking only the values of $\b{a}$ on the coordinates in $S$. We say that $S$ is \textit{shattered} by $\m{C}$ \cite{AlonSpencer}, if the projection multiset
\begin{align*}
P^+_S(\m{C})\triangleq\left\{\b{c}(S): \b{c}\in\m{C} \right\} \quad \text{with multiplicities}
\end{align*}
of $\m{C}$ on $S$ contains all $2^{|S|}$ binary vectors of length $|S|$.\footnote{Taking the multiplicities into account in the definition of the projection multiset is not necessary here, but will become important in the sequel.} A codebook $\m{C}$ is said to be \textit{systematic} if it is shattered by some $S\subseteq [n]$ of cardinality $ \log |\m{C}|$. Weldon proved the following.
\begin{theorem}[Weldon~\cite{Weldon78}]\label{thm:weldon}
  If $\m{C}_1$ is systematic and ($\m{C}_1,\m{C}_2$) form a zero-error codebook pair, then $R_2 \leq (1-R_1)\log{3}$.
\end{theorem}
\begin{proof}
Let $S$ be a set of cardinality $  nR_1  $ that is shattered by $\m{C}_1$. For every $\b{c}_2\in\m{C}_2$, there exists a $\b{c}_1\in \m{C}_1$ such that $\b{c}_1$ and $\b{c}_2$ are an \textit{$S$-complement} pair, i.e.,
  \begin{align}\label{eq:all1}
\b{c}_1(S)+ \b{c}_2(S) = \b{1}_{|S|},
  \end{align}
where $\b{1}_{m}$ denotes a vector of $1$s of length $m$.
Hence, there are at least $2^{nR_2}$ such $S$-complement pairs. By the assumption that ($\m{C}_1,\m{C}_2$) form a zero-error codebook pair, $\b{c}_1(\overline{S})+ \b{c}_2(\overline{S})$ must be distinct for all $S$-complement pairs. Therefore, the number of such pairs cannot be larger than $3^{|\overline{S}|} = 3^{n(1-R_1)}$, and the theorem follows.
\end{proof}

For example, if $\m{C}_1$ is systematic and $R_2=1$, then the theorem implies that $R_1 \leq 0.37$. This strong bound is a consequence of the restriction to a systematic codebook. However, we note that the only property used in the proof is the existence of a large shattered set. Hence, any lower bound on the size of a maximal shattered set in a general codebook $\m{C}_1$ would lead to a similar result. The cardinality of the maximal set shattered by a code $\m{C}\subseteq\{0,1\}^n$ is referred to in the machine-learning literature as its \emph{Vapnik-Chervonenkis dimension}, or \emph{VC-dimension}. The Sauer-Perles-Shelah lemma provides a lower bound on the VC-dimension of a code.
\begin{lemma}[Sauer-Perles-Shelah Lemma \cite{AlonSpencer}]
If the cardinality of the maximal subset shattered by the codebook $\m{C}\subseteq\{0,1\}^n$ is $d$, then 
\begin{align*}
|\m{C}| \leq \sum_{k=0}^d{n \choose k }. 
\end{align*}
\end{lemma}
\begin{remark}
  It is easy to see that this bound is attained with equality if $\m{C}$ is a $n$-Hamming ball of radius $d$.
\end{remark}
\begin{corollary}
Let $\varepsilon>0$. If $|\m{C}| = 2^{n(R+\varepsilon)}$ then for any $n$ large enough, $\m{C}$ shatters a set $S\subseteq [n]$ with $|S|\geq n h^{-1}(R)$.
\end{corollary}

Plugging the above into Weldon's argument yields:
\begin{proposition}
If ($\m{C}_1,\m{C}_2$) form a zero-error codebook pair, then $R_2 \leq (1-h^{-1}(R_1))\log{3}$.
\end{proposition}

Unfortunately, this bound is trivial since for any $R_1$, we have that $R_1+(1-h^{-1}(R_1))\log{3} > \tfrac{3}{2}$. This stems from two main weaknesses. First, we have taken the worst case assumption that each codeword $\b{c}_2\in \m{C}_2$ has only one codeword $\b{c}_1\in\m{C}_1$ such that $\b{c}_1$ and $\b{c}_2$ are $S$-complement, where $S$ is a shattered set in $\m{C}_1$. Second, bounding the number of $S$-complement pairs by $3^{|\overline{S}|}$ may be loose, as it ignores the sumset structure. In the next two subsections, we provide the technical tools to handle each of these weaknesses, and apply them to prove the theorem in the subsection that follows.

\subsection{A Soft  Sauer-Perles-Shelah Lemma}

Let $\m{C}\subseteq\{0,1\}^n$ be a codebook and let $S\subseteq [n]$ be a subset of coordinates. We say that $S$ is \textit{$k$-shattered} by $\m{C}$, if the projection multiset $P_S^+(\m{C})$ of $\m{C}$ on $S$ contains all binary vectors in $\{0,1\}^{|S|}$ each with multiplicity of at least $k$. For $k=1$, this definition reduces to the regular definition of a shattered set.

The proof of the following lemma is given in Section~\ref{sec:proof_soft_sauer}.
\begin{lemma}\label{lem:soft_sauer}
If the cardinality of the maximal subset that is $k$-shattered by the codebook $\m{C}\subseteq\{0,1\}^n$ is $d-1$, then
  \begin{align*}
     |\m{C}|\leq \sum_{t=1}^{t^*}{n\choose t}  + {n\choose t^*}\sum_{t=t^*+1}^{n}\frac{{t^*\choose d}}{{t \choose d}}
  \end{align*}
where $t^*$ is the smallest integer $t$ satisfying ${n-d \choose t-d} \geq k$ if such an integer exists, and $t^*=n$ otherwise.
\end{lemma}
\begin{remark}
Note that if $k={n-d \choose t^*-d}$ for some $t^*$, then our bound is tight for a $n$-Hamming ball of radius $t^*$, up to a multiplicative gap of $O(n/d)$. This coincides with the Sauer-Perles-Shelah Lemma for $k=1$ (and $t^*=d$), up to the aforementioned multiplicative factor. Since we are only interested in exponential behavior, no attempt has been made to reduce this gap.
\end{remark}

\begin{corollary}\label{cor:soft_sauer}
Let $\varepsilon>0$. If $|\m{C}| = 2^{n(R+\varepsilon)}$ then for any $0\leq \alpha \leq h^{-1}(R)$ and any $n$ large enough, there exists a set $S\subseteq[n]$ with $|S|\geq n\alpha$ that is $\;2^{n\beta}$-shattered by $\m{C}$, where
  \begin{align}\label{eq:beta}
    \beta = (1-\alpha)\cdot h\left(\frac{h^{-1}(R)-\alpha}{1-\alpha}\right)
\end{align}
\end{corollary}
\begin{proof}
Let $0\leq \alpha \leq h^{-1}(R)$ and assume to the contrary that no subset of size $d=n\alpha$ is $2^{n\beta}$-shattered by $\m{C}$. Denote $t^*=\gamma_n n$, and write
\begin{align*}
\frac{1}{n}\log{n-d\choose t^*-d} &= \frac{n-d}{n} \left(h\left(\frac{t^*-d}{n-d}\right) + o(1)\right) \\
&= (1-\alpha + o(1))h\left( \frac{\gamma_n - \alpha}{1-\alpha}\right)
\end{align*}
We can set $\gamma_n$ to the minimal value guaranteeing that the above is at least $\beta$, which is $\gamma_n =\alpha+(1-\alpha)h^{-1}\left(\frac{\beta}{1-\alpha}\right) + o(1)$. Invoking Lemma \ref{lem:soft_sauer}, it must then be that $|\m{C}| \leq 2^{n(h(\gamma_n) + o(1))} = 2^{n(R+o(1))}$, contradicting the assumption.
\end{proof}

\subsection{The Binary Adder Channel with an Additional Common Message}
In the Weldon-type arguments mentioned above, the number of $S$-complement pairs was bounded by $3^{|\overline{S}|}$, thereby ignoring the sumset structure. As we shall see in the next subsection, this structure can be accounted for by partitioning each codebook according to its projection on $S$, which naturally gives rise to a zero-error communication problem with an additional common message of rate at most $|S|/|\overline{S}|$. Upper bounding the corresponding admissible sum-rate in this new setup can in turn be translated into an upper bound on the number of $S$-complement pairs in our original setup.

More precisely, assume that there are three messages $W_i\in[2^{nr_i}]$, $i=0,1,2$, to be conveyed to the receiver over the binary adder channel, where the first user has access to the messages $(W_0,W_1)$ and the second user has access to the messages $(W_0,W_2)$. The Shannon capacity region for this problem was found by Slepian and Wolf~\cite{SW73} to be the set of all rate triplets satisfying
\begin{align}\label{eq:sw}
r_1&\leq H(X_1|U),\nonumber\\
r_2&\leq H(X_2|U), \nonumber\\
r_1+r_2 &\leq H(X_1+X_2|U),\nonumber\\
r_0+r_1+r_2 &\leq H(X_1+X_2)
\end{align}
for some $P_{U,X_1,X_2}=P_U P_{X_1|U} P_{X_2|U}$, where $X_1$ and $X_2$ are binary random variables and the random variable $U$ has a finite support.

A coding scheme for this problem consists of a \textit{system} $\m{V}$, which is a set of codebook pairs $\{\m{C}_{1,i},\m{C}_{2,i}\}_{i=1}^{M_0}$, where each $\m{C}_{1,i}$ (resp.  $\m{C}_{2,i}$) is a codebook in $\{0,1\}^n$ with fixed cardinality $|\m{C}_{1,i}| = M_1$ (resp. $|\m{C}_{2,i}| = M_2$). We say that $\m{V}$ is a \textit{zero-error system} if each pair $(\m{C}_{1,i},\m{C}_{2,i})$ is a zero-error codebook pair, and the sumsets $\m{C}_{1,i}+ \m{C}_{2,i}$ are mutually disjoint. A triplet $(r_0,r_1,r_2)$ is called admissible if there exists a zero-error system $\m{V}$ with $M_\ell = 2^{n(r_\ell+o(1))}$ for $\ell\in\{0,1,2\}$.

Clearly, any admissible triplet must satisfy~\eqref{eq:sw}. The bounds we obtain in this subsection are based on outer bounding this latter region. More specifically, as will become clear in the next subsection, our goal is to upper bound the maximal sum of admissible rates $r_0+r_1+r_2$ as a function of $r_0$ and $r_1$. Although the bounds in~\eqref{eq:sw} are given in a single-letter form, in order to guarantee the inadmissibility of a rate triplet, one must go over all valid distributions $P_{U,X_1,X_2}$. While it is not difficult to show that for our needs there is no loss of generality in considering only random variables $U$ with cardinality no greater than $3$, the number of remaining parameters makes the evaluation of~\eqref{eq:sw} within a satisfactory resolution infeasible for a brute-force grid search. Instead, the following lemma provides an analytic upper bound on the sum-capacity as a function of $r_0$ and $r_1$, in terms of the solution to a single-parameter optimization problem. The proof is omitted due to space limitations, but can be found in the full version of this paper~\cite{os14}.

%For $r_0=0$, the problem coincides with the standard binary adder problem, for which $r_0+r_1+r_2\leq \tfrac{3}{2}$ follows from~\eqref{eq:ShannonCap}. It is also easy to see that for a large enough value of $r_0$, the sum $r_0+r_1+r_2 = \log{3}$ is admissible. For example, let $\m{C}_0 = \{\b{c}_{0,1},\ldots,\b{c}_{0,M_0}\}$ be the set of all vectors in $\{0,1\}^n$ whose Hamming weight is exactly $2n/3$, and identify each pair $\{\m{C}_{1,i},\m{C}_{2,i}\}$ in the system $\m{U}$ with one of the these vectors. Let $\m{C}_{1,i} = \{\b{c}_{0,i}\}$ and $\m{C}_{2,i}=\{\b{c}\in\{0,1\}^n: F\subseteq F_{0,i}\}$. Clearly, each pair $(\m{F}_{1,i},\m{F}_{2,i})$ is multiset-union-free, and moreover, the families of multisets $\m{F}_{1,i}\uplus \m{F}_{2,i}$ as defined above are disjoint, as exactly all the elements of $F_{0,i}$ participate in each corresponding multiset family. For this construction, $r_0=\tfrac{1}{n}\log {n \choose 2n/3} \approx h(\tfrac{1}{3})$, $r_1=0$ and $r_2 = \tfrac{2}{3}$, hence in the limit of large $n$ this construction yields $r_0+r_1+r_2 = \log{3}$.The next lemma refines these observations by upper bounding admissible sums $r_0+r_1+r_2$ between $\tfrac{3}{2}$ and $\log{3}$, as a function of $r_0$ and $r_1$. The proof appears in Section \ref{sec:proof_sw}.

\begin{lemma}\label{lem:sw}
  Let $L(\eta)$ and $J(p,\eta)$ be as defined in~\eqref{eq:L} and~\eqref{eq:J}. If $(r_0,r_1,r_2)$ is admissible, then
\begin{align*}
  r_0+r_1+r_2 \leq \max_{h^{-1}(r_1)\leq \eta\leq \frac{1}{2}} \min\{L(\eta),\, J(h^{-1}(r_1),\eta)+r_0\}
\end{align*}
\end{lemma}

\begin{remark}
  Note that it can be shown that the maximization can be further restricted to $h^{-1}(r_1)\star h^{-1}(r_2)\leq \eta \leq \frac{1}{2}$. This however is not useful for our purposes.
\end{remark}

The following lemma is not necessary for the proof of Theorem~\ref{thm:main}, but may be of independent interest.
\begin{lemma}\label{lem:sumsw}
The maximal sum of achievable rates (for a vanishing error probability) over the binary adder channel with an additional common message, as a function of the rate of the common message rate $r_0$, is
\begin{align}
  r_0+r_1+r_2 = &\max_{0\leq \eta\leq \frac{1}{2}} \min\{h(\eta)+1-\eta,\nonumber\\
  &2h\left(\tfrac{1}{2}\left(1-\sqrt{1-2\eta}\right)\right)-\eta+r_0\}\label{eq:swsumcap}
\end{align}
\end{lemma}
\begin{proof}[Proof sketch]
The upper bound on $r_0+r_1+r_2$ follows as a corollary of Lemma~\ref{lem:sw}, by noting that for any $0\leq r_1\leq 1$ we have $J(h^{-1}(r_1),\eta)\leq2h\left(\tfrac{1}{2}\left(1-\sqrt{1-2\eta}\right)\right)-\eta$. To see that the right hand side of ~\eqref{eq:swsumcap} is achievable, let $\eta^*$ be the maximizer of~\eqref{eq:swsumcap} and evaluate the entropies in~\eqref{eq:sw} with the following distribution:
\begin{align}\label{eq:opt_dist}
\nonumber &X_1 = U \oplus Z_1, \; X_2 = U \oplus Z_2 \\
&U\sim \textrm{Bern}\left(\frac{1}{2}\right),\; Z_1\sim\textrm{Bern}(p^*), \; Z_2\sim\textrm{Bern}(p^*)
\end{align}
where $U,Z_1,Z_2$ are mutually independent, and $p^*\leq \tfrac{1}{2}$ satisfies $p^*\star p^* = \eta^*$, i.e., $p^* =  \frac{1}{2}(1-\sqrt{1-2\eta^*})$. \end{proof}

\subsection{Putting it Together}
We are now in a position to prove Theorem \ref{thm:main}. Let $(\m{C}_1,\m{C}_2)$ be a zero-error codebook pair of cardinalities $2^{nR_1}$ and $2^{nR_2}$ respectively. Given this pair, we use Corollary \ref{cor:soft_sauer} to  construct a zero-error system with certain cardinalities, and then apply Lemma \ref{lem:sw} to obtain constraints on that system.

By Corollary \ref{cor:soft_sauer}, for any $\alpha<h^{-1}(R_1)$ there exists a subset of coordinates $S\subset [n]$ of cardinality $n\alpha$ that is $2^{n\beta}$-shattered by $\m{C}_1$, where $\beta$ is given in~\eqref{eq:beta}, all up to an $o(1)$ term. Let $\m{C}_0$ be the family of all binary vectors of length $|S|$, and for any $\b{g}\in\m{C}_0$ let $\m{C}_{1,\b{g}} = \{\bc\in\m{C}_1 : \b{c}(S)= \b{g}\}$. Define $\m{C}_{2,\b{g}}$ similarly, and note that $\{\m{C}_{j,\b{g}}\}_{\b{g}\in\m{C}_0}$ is a partition of $\m{C}_j$ for each $j\in\{1,2\}$.

By construction, $|\m{C}_{1,\b{g}}| \geq 2^{n\beta}$. We can therefore arbitrarily choose $\wt{\m{C}}_{1,\b{g}}\subseteq \m{C}_{1,\b{g}}$ such that $|\wt{\m{C}}_{1,\b{g}}| = 2^{n\beta}$. For each $\b{g}$ with $|\m{C}_{2,\b{g}}|>0$, arbitrarily choose $\wt{\m{C}}_{2,\b{g}}\subseteq \m{C}_{2,\b{g}}$ such that $\log|\wt{\m{C}}_{2,\b{g}}| = \lfloor\log|\m{C}_{2,\b{g}}|\rfloor$. Note that this guarantees that $|\wt{\m{C}}_{2,\b{g}}| = 2^k$ for some integer $0\leq k\leq nR_2$, and that $|\wt{\m{C}}_{2,\b{g}}| \geq |\m{C}_{2,\b{g}}|/2$. Moreover, there must exist an integer $k'$ with the property that the union of all $\wt{\m{C}}_{2,\b{g}}$ of cardinality $2^{k'}$ contains at least $\tfrac{1}{2(nR_2+1)}2^{nR_2}$ vectors. Let $\m{G}$ be the set of all $\b{g}\in\m{C}_0$ that correspond to this $k'$, and note that by construction $|\m{G}|= 2^{n\alpha'}$ for some $\alpha'\leq \alpha$. Moreover,
\begin{align*}
|\wt{\m{C}}_{2,\b{g}}| = 2^{k'} \geq  \tfrac{1}{2(nR_2+1)}2^{n(R_2-\alpha')}  
\end{align*}
for all $\b{g}\in\m{G}$.

Let $\overline{\b{g}} =\b{g}\oplus\b{1}_{|S|}$ be the binary complement of $\b{g}$, and define the system $\m{V} = \{(\wt{\m{C}}_{1,\overline{\b{g}}},\wt{\m{C}}_{2,\b{g}})\}_{\b{g}\in\m{G}}$. Since the original $\m{C}_1$ and $\m{C}_2$ form a zero-error codebook pair, then $\m{V}$ is trivially a zero-error system. Moreover, since any $\b{c}_1\in \wt{\m{C}}_{1,\overline{\b{g}}}$ and $\b{c}_2\in \wt{\m{C}}_{2,\b{g}}$ are an $S$-complement pair~\eqref{eq:all1}, the projection 
\begin{align*}
\m{V}_{\overline{S}} \triangleq  \{(P^+_{\overline{S}}(\wt{\m{C}}_{1,\overline{\b{g}}}), P^+_{\overline{S}}(\wt{\m{C}}_{2,\b{g}}))\}_{\b{g}\in\m{G}}  
\end{align*}
of $\m{V}$ onto $\overline{S}$ is also a zero-error system, over $|\overline{S}| = n(1-\alpha)$ coordinates.

We have thus shown that given a zero-error codebook pair over $n$ coordinates with cardinalities $2^{nR_1}$ and $2^{nR_2}$, we can construct a zero-error system $\m{V}_{\overline{S}}$ over $m= n(1-\alpha)$ coordinates with cardinalities $M_0 =  2^{mr_0}$, $M_1 = 2^{mr_1}$ and $M_2 = 2^{m\left(r_2+o(1)\right)}$, where 
\begin{align*}
r_0=\frac{\alpha'}{1-\alpha}, \quad r_1=\frac{\beta}{1-\alpha},\quad r_2=\frac{R_2-\alpha'}{1-\alpha}  
\end{align*}
Thus for this system $r_0+r_1+r_2 = \frac{R_2+\beta}{1-\alpha}$, and by Lemma \ref{lem:sw}, recalling that $\alpha'\leq \alpha$, we have that
\begin{align*}
\frac{R_2+\beta}{1-\alpha}\leq \max_{h^{-1}\left(\frac{\beta}{1-\alpha}\right)\leq \eta\leq \frac{1}{2}} &\min\bigg\{L(\eta),\nonumber\\
&J\left(h^{-1}\left(\frac{\beta}{1-\alpha}\right),\eta\right)+\frac{\alpha}{1-\alpha}\bigg\}
\end{align*}

The theorem now follows by substituting $\beta$ from Corollary \ref{cor:soft_sauer}, and noting that the above inequality holds for any $0\leq \alpha \leq h^{-1}(R_1)$.

\section{Proof of Lemma \ref{lem:soft_sauer}}\label{sec:proof_soft_sauer}
For the purpose of the proof, it will be convenient to represent any binary vector $\b{c}\in\{0,1\}^n$ by a subset of $F\subseteq [n]$ that contains the indices of the coordinates where $\b{c}$ equals $1$. Accordingly, any codebook $\m{C}\subseteq \{0,1\}^n$ can be represented by the corresponding family $\m{F}$ of subsets of $[n]$. Similarly, the multiset projection $P_S^+(\m{C})$ of $\m{C}$ on $S$ corresponds to
\begin{align*}
P_S^+(\m{F}) \triangleq \left\{F\cap S: F\in\m{F} \right\} \qquad \text{with multiplicities}
\end{align*}
and $S$ is $k$-shattered by $\m{C}$ (equivalently by $\m{F}$) means that $P_S^+(\m{F})$ contains each subset of $S$ with multiplicity at least $k$.

Let $\m{C}$ be a codebook and let $\m{F}$ be the corresponding family of subsets on $[n]$. We start by applying the shifting argument introduced in \cite{Alon83} on $\m{F}$, to construct another family $\m{G}$ of the same cardinality, such that if $S$ is $k$-shattered by $\m{G}$ then it is also $k$-shattered by $\m{F}$. Furthermore, $\m{G}$ will be \textit{monotone}, i.e., will have the property that if $G\in\m{G}$ then all subsets of $G$ are in $\m{G}$.

Set $\m{G}=\m{F}$. If $\m{G}$ is already monotone, we are done. Otherwise there exists some $i\in[n]$ such that  the set 
\begin{align*}
\wt{\m{G}}_i \triangleq  \{G\in\m{G} : i\in G,\,G\setminus \{i\}\not\in\m{G}\}  
\end{align*}
is not empty. Update $\m{G}$ according to the rule:
\begin{align}\label{eq:update}
\m{G} \leftarrow \left(\m{G}\setminus \wt{\m{G}}_i\right) \cup \left(\wt{\m{G}}_i - i\right)
\end{align}
where $\wt{\m{G}}_i - i$ is the family of subsets obtained from $\wt{\m{G}}_i$ by removing the element $i$ from each subset. The process continues until $\m{G}$ is monotone, and is clearly guaranteed to terminate in finite time. By construction, $|\m{G}|=|\m{F}|$.

We now show that if $S$ is $k$-shattered by $\m{G}$ then it is also $k$-shattered by $\m{F}$. Let $\m{G}'$ be the family of subsets before the operation~\eqref{eq:update} on some element $i$, and let $\m{G}$ be the family obtained after that operation. Suppose $S$ is $k$-shattered by $\m{G}$. It now suffices to show that $S$ is also $k$-shattered by $\m{G}'$. If $i\not\in S$ then clearly $P_S^+(\m{G}) = P_S^+(\m{G}')$, hence this does not affect the $k$-shatterdness of $S$. Suppose $i\in S$, and let 
\begin{align*}
\m{G}_i  \triangleq \{G\in\m{G} : i\in G\}.
\end{align*}
Then $\m{G}_i\subseteq \m{G}'$ since the update rule~\eqref{eq:update} does not add elements to subsets. Since $\m{G}$ $k$-shatters $S$, then every subset of $S$ that contains $i$ has multiplicity at least $k$ in $P_S^+(\m{G}_i)\subseteq P_S^+(\m{G}')$. Recalling that $\m{G}_i\subseteq \m{G}\cap\m{G}'$, we have that $\m{G}_i - i \subseteq \m{G}'$ since otherwise some replacement would have occurred in~\eqref{eq:update}. Since $\m{G}$ $k$-shatters $S$, then every subset of $S$ that does not contain $i$ has multiplicity at least $k$ in $P_S^+(\m{G}_i - i)\subseteq P_S^+(\m{G}')$.

% Define  $\overline{\m{F}_i''} = \m{F}''\setminus \m{F}_i''$, and let $\overline{\m{F}_i''} + i$ be obtained from $\overline{\m{F}_i''}$ by adding $i$ to each subset. Then by $k$-shatterdness we have that there must exist a family $\m{F}^*\subseteq \m{F}''$ such that $P_S(\m{F}^*) = P_S(\overline{\m{F}_i''} + i)$ and that every element in the multiset $P_S^+(\m{F}^*)$ has multiplicity at least $k$. Now note that $\m{F}^* - i\subseteq\m{F}'$, as otherwise some  replacement would have occurred. Thus, $\m{F}_i''\cup (\m{F}^*-i)$ is contained in $\m{F}'$, and also $k$-shatters $S$.

The Lemma now follows directly from the next proposition.

\begin{proposition}
  If $\m{G}$ is a monotone family of subsets of $[n]$ with the property that no subset of cardinality $d$ is $k$-shattered by $\m{G}$, then
  \begin{align*}
     |\m{G}|\leq \sum_{t=1}^{t^*}{n\choose t}  + {n\choose t^*}\sum_{t=t^*+1}^{n}\frac{{t^*\choose d}}{{t \choose d}}
  \end{align*}
where $t^*$ is the smallest integer $t$ satisfying ${n-d \choose t-d} \geq k$ if such an integer exists, and $t^*=n$ otherwise.
\end{proposition}

\begin{proof}
Let $\m{G}_t$ denote the family of all subsets in $\m{G}$ with cardinality $t$. For $t\geq d$, every $G\in\m{G}_t$ has exactly ${t \choose d}$ subsets of cardinality $d$. There is a total of ${n \choose d}$ subsets of cardinality $d$. Hence by a simple counting argument there must exist at least one subset $S$ of cardinality $d$, that is a subset of no less than $|\m{G}_t|{t\choose d}/{n \choose d}$ subsets in $\m{G}_t$. Recalling that $\m{G}$ is monotone, this implies that $S$ is $|\m{G}_t|{t\choose d}/{n \choose d}$-shattered by $\m{G}$. By our assumption, it must be that
\begin{align*}
  \frac{{t\choose d}|\m{G}_t|}{{n \choose d}} < k,\quad t=d,\ldots,n
\end{align*}
On the other hand, $|\m{G}_t|\leq {n \choose t}$, and therefore
\begin{align*}
  |\m{G}_t| \leq \min\left\{{n \choose t}, \frac{{n\choose d}k}{{t \choose d}}\right\}, \quad t=d,\ldots,n
\end{align*}
Summing over $t$ we get
\begin{align}\label{eq:Gbound}
  |\m{G}| &= \sum_{t=1}^n|\m{G}_t| \leq  \sum_{t=1}^{d-1}{n\choose t} + \sum_{t=d}^n\min\left\{{n \choose t}, \frac{{n\choose d}k}{{t \choose d}}\right\}
\end{align}
Let $t^*$ be the smallest integer $t$ such that ${n \choose t} \geq \frac{{n\choose d}k}{{t \choose d}}$ if such an integer exists. If no such integer $t$ exists, set $t^* = n$. Then
\begin{align*}
  |\m{G}| &\leq \sum_{t=1}^{t^*}{n\choose t}  + \sum_{t=t^*+1}^{n}\frac{{n\choose d}k}{{t^*\choose d}}\cdot\frac{{t^*\choose d}}{{t \choose d}} \\
&\leq \sum_{t=1}^{t^*}{n\choose t}  + {n\choose t^*}\sum_{t=t^*+1}^{n}\frac{{t^*\choose d}}{{t \choose d}} \end{align*}

To complete the proof, note that for any $d\leq t\leq n$ we have ${n\choose t}{t\choose d} = {n\choose d}{n-d\choose t-d}$, hence $t^*$ is the smallest integer $t$ satisfying ${n-d\choose t-d} \geq k$ if such an integer exists, and otherwise $t^*=n$.
\end{proof}

\section{Discussion}
Given a zero-error codebook pair $\m{C}_1,\m{C}_2\subseteq\{0,1\}^n$ with cardinalities $2^{nR_1}$ and $2^{nR_2}$ respectively, our bounding technique was based on a procedure for constructing a zero-error system $\m{V}$ with dimension $(1-\alpha) n$. This was achieved by proving the existence of a subset $S\subset [n]$ of cardinality $\alpha n$, such that the sumset of the projection multisets of each codebook on $S$, i.e., $P^+_S(\m{C}_1)+ P^+_S(\m{C}_2)$ has a member $\b{v}\in\{0,1,2\}^{|S|}$ with a large number of occurrences, say $2^{n\rho}$. This in turn implied that $r_0+r_1+r_2$ for the system is at least $\rho/(1-\alpha)$. To lower bound $\rho$ as a function of $\alpha$ and the cardinalities of the original codebooks, we introduced the soft Sauer-Perles-Shelah Lemma, which enabled us to bound the number of occurrences of the vector $\b{v}=\b{1}_{|S|}$. This lemma offered the additional benefit of a lower bound on $r_1$. We note in passing that the bound obtained on $R_2$ as a function of $R_1$ outperforms previous results even without incorporating the constraint on $r_1$. We suspect that better bounds on $\rho$ can be obtained, possibly for $\b{v}$ other than $\b{1}_{|S|}$.

\bibliographystyle{IEEEtran}
\bibliography{OrBib2}

\end{document}